\documentclass[12pt]{article}
\pdfoutput=1
\usepackage{amsfonts,amsmath,amssymb,amsthm,graphicx}
\usepackage{enumerate,enumitem}
\usepackage{setspace} 
\usepackage{hyperref}
\usepackage[capitalize]{cleveref}
\usepackage{ifthen}
\usepackage[margin=1.25in]{geometry}

\usepackage[font=footnotesize,labelfont=bf]{caption}

\usepackage{natbib}
\usepackage{tikz}
\usetikzlibrary{decorations.pathreplacing}
\usetikzlibrary{patterns}
\usetikzlibrary{plotmarks}

\usepackage{subfigure}

\usepackage{bbm}

\theoremstyle{plain}
\newtheorem{theorem}{Theorem}
\newtheorem{lemma}{Lemma}
\newtheorem{claim}{Claim}
\newtheorem{proposition}{Proposition}
\newtheorem{corollary}{Corollary}
\newtheorem*{problem-nonum}{Problem}

\newtheorem{definition}{Definition}

\usepackage{color-edits}
\addauthor{yl}{blue}    
\addauthor{bx}{red}

\title{\Large Learning and Communication Towards Unanimous Consent\thanks{We thank Nageeb Ali, Arjada Bardhi, Marco Battaglini, Dirk Bergemann, Xiaoyu Cheng, Joyee Deb, Theo Durandard, Mira Frick, Tan Gan, Johannes H\"orner, Ryota Iijima, Fei Li, Xiao Lin, Bart Lipman, Yichuan Lou, Lucas Maestri, Harry Pei, Allen Vong, Thomas Wiseman, and audiences at Yale, CUHK-Shenzhen, INFORMS 2022, ICGT 2023, GAMES 2024, SEA 2024, SAET 2025, and Midwest Theory Conference 2025 for helpful comments and suggestions. Part of the work was done when Yingkai Li was a postdoc at Yale University under the support of Sloan Research Fellowship FG-2019-12378. Yingkai Li also thanks NUS Start-up grant for financial support.}}

\author{Yingkai Li\thanks{Department of Economics, National University of Singapore.
Email: \texttt{yk.li@nus.edu.sg}} 
\and Boli Xu\thanks{Department of Economics, University of Iowa.
Email: \texttt{boli-xu@uiowa.edu}}} 

\begin{document}
\newcommand{\state}{\theta}
\newcommand{\states}{\Theta}

\newcommand{\type}{\lambda}
\newcommand{\types}{\Lambda}

\newcommand{\test}{\sigma}
\newcommand{\tests}{\mathcal{T}}
\newcommand{\menu}{T}
\newcommand{\signals}{S}
\newcommand{\signal}{s}

\newcommand{\talkmessages}{M}
\newcommand{\talkmessage}{m}

\newcommand{\val}{v}
\newcommand{\util}{u}
\newcommand{\intP}{U}
\newcommand{\intutil}{V}

\newcommand{\cutoffstate}{\state}
\newcommand{\cutoffposterior}{\posterior^*}

\newcommand{\action}{a}
\newcommand{\pay}{b}

\newcommand{\alloc}{x}
\newcommand{\virtual}{\Phi}

\newcommand{\posteriormean}{\mu}

\newcommand{\typedist}{F}
\newcommand{\typedensity}{f}

\newcommand{\statedist}{G}
\newcommand{\statedensity}{g}

\newcommand{\principalPosterior}{\psi_p}
\newcommand{\posterior}{\psi}
\newcommand{\principalPosteriorl}{\psi_{p,l}}
\newcommand{\principalPosteriorh}{\psi_{p,h}}
\newcommand{\principalPosteriorQ}{\psi_q}
\newcommand{\agentPosterior}{\psi_a}

\newcommand{\setsize}[1]{{\left|#1\right|}}

\newcommand{\floor}[1]{
{\lfloor {#1} \rfloor}
}
\newcommand{\bigfloor}[1]{
{\left\lfloor {#1} \right\rfloor}
}

%
%
\newcommand{\given}{\,|\,}

\newcommand{\prob}[2][]{\text{\bf Pr}\ifthenelse{\not\equal{}{#1}}{_{#1}}{}\!\left[{\def\givenn{\middle|}#2}\right]}
\newcommand{\expect}[2][]{\text{\bf E}\ifthenelse{\not\equal{}{#1}}{_{#1}}{}\!\left[{\def\givenn{\middle|}#2}\right]}

\newcommand{\tparen}{\big}
\newcommand{\tprob}[2][]{\text{\bf Pr}\ifthenelse{\not\equal{}{#1}}{_{#1}}{}\tparen[{\def\given{\tparen|}#2}\tparen]}
\newcommand{\texpect}[2][]{\text{\bf E}\ifthenelse{\not\equal{}{#1}}{_{#1}}{}\tparen[{\def\given{\tparen|}#2}\tparen]}

\newcommand{\sprob}[2][]{\text{\bf Pr}\ifthenelse{\not\equal{}{#1}}{_{#1}}{}[#2]}
\newcommand{\sexpect}[2][]{\text{\bf E}\ifthenelse{\not\equal{}{#1}}{_{#1}}{}[#2]}

\newcommand{\rbr}[1]{\left(\,#1\,\right)}
\newcommand{\sbr}[1]{\left[\,#1\,\right]}
\newcommand{\cbr}[1]{\left\{\,#1\,\right\}}

\newcommand{\suchthat}{\,:\,}

\newcommand{\partialx}[2][]{{\tfrac{\partial #1}{\partial #2}}}
\newcommand{\nicepartialx}[2][]{{\nicefrac{\partial #1}{\partial #2}}}
\newcommand{\dd}{\,{\mathrm d}}
\newcommand{\ddx}[2][]{{\tfrac{\dd #1}{\dd #2}}}
\newcommand{\niceddx}[2][]{{\nicefrac{\dd #1}{\dd #2}}}
\newcommand{\grad}{\nabla}

\newcommand{\symdiff}{\triangle}

\newcommand{\abs}[1]{\left| #1 \right |}

\newcommand{\reals}{{\mathbb R}}
\newcommand{\posreals}{\reals_+}
\newcommand{\indicate}[1]{\mathbf{1}\left[\,#1\,\right]}

\newcommand{\primed}{^{\dagger}}

\date{}
\maketitle

\begin{abstract}

\noindent A principal and an agent can launch a project under unanimous consent. Their individual payoffs from the project depend on an underlying state, and the agent privately knows his own preference. The principal can conduct a test to learn about the state and then communicate with the agent, but has limited commitment, as she may misreport her findings. We show that limited commitment makes binary tests optimal. Moreover, when players' preferences are positively aligned, the optimal test is a \textit{threshold test}. When their preferences are negatively aligned, the optimal test is either an \textit{interval test} or a \textit{tail test}, depending on the agent's relative risk attitude. Additionally, the principal can benefit from screening the agent through a menu of tests, which admits a simple structure regardless of the complexity of the agent's type space.\\


\vspace{6pt}

\vspace{3pt}

\end{abstract}

\thispagestyle{empty}
\newpage
\setcounter{page}{1}

\begin{spacing}{1.5}

\section{Introduction}
\label{sec:intro}

In many situations of collective decision-making, a certain action can be taken if all parties consent to it, but it is unclear ex-ante how each party will gain or lose from the action. To facilitate collaboration, one party can learn about the action and communicate her findings to the other parties. However, given the potential divergence of interests, the information acquirer often lacks the credibility to truthfully reveal what she has learned, making it challenging to reach unanimous consent. 

This issue arises in a variety of contexts. For instance, a company's CEO may commission a consulting report to persuade the board to launch a business plan with uncertain profitability. A political party may use a think-tank study to convince the opposing party to support a policy whose effects are ex-ante unknown. A manager may use a feasibility study to convince a worker to undertake a novel task of unknown difficulty. A researcher may present a proposal to persuade potential collaborators to embark on a project with limited knowledge about its chances of success. Unfortunately, in all of these examples, the information acquirer has both the opportunity and the incentive to distort or manipulate the information generated and presented.\footnote{For instance, it is widely recognized that some consulting firms may tailor their findings to clients' predetermined objectives.}

Given this credibility issue, how can a party convincingly acquire and communicate information to persuade other parties? In this paper, we study this question by developing a model in which a principal and an agent can launch a project under unanimous consent. Their individual payoffs from the project are determined by a state drawn from a continuous distribution. In addition, the agent privately knows how his payoff depends on the state, reflecting the realistic possibility that the principal is imperfectly informed about the agent's preference.\footnote{For example, a CEO may not know which level of profitability the board deems sufficient to warrant investment, a political party may not know how the opposing party evaluates a policy, a manager may not know a worker's capability to handle a difficult task, and a researcher may not know a potential collaborator's current bandwidth to embark on a new project.} The game begins with the principal publicly adopting a test to learn about the payoff-relevant state. The result of the test is privately revealed to the principal, who then sends a cheap-talk message to the agent conveying her findings, which, crucially, need not be truthful. After the communication, the principal decides whether to propose or forgo the project. If she proposes the project, the agent then decides whether to accept or reject it.

\Cref{sec:general} presents the general results. The principal in our model has \textit{limited commitment} in the sense that she has ex-ante commitment to an information structure but \textit{no} ex-post commitment to truthful reporting of the acquired information. This feature distinguishes our setting from standard information design models, where the principal has full commitment. Allowing full commitment would render our model intractable, particularly due to the combination of a privately informed agent, a continuous state space, and non-linear payoff functions of the players.\footnote{Existing information design models featuring a privately informed agent and a continuous state space focus on linear persuasion problems, in which the players' payoffs are linear in the state (see, e.g., \citet{kolotilin2022censorship} and \citet{candogan2023optimal}). By contrast, our model allows for arbitrary payoff functions, enabling a broader range of applications.} Somewhat surprisingly, the principal's limited commitment restores the tractability of the model and yields simplicity in the results. \Cref{prop:binary} shows that it is without loss of optimality for the principal to adopt a binary test that returns either a \emph{proposal signal}, leading her to propose the project, or a \emph{null signal}, inducing her to forgo it. This result follows from the requirement to make the test \textit{trustworthy}, in the sense that the principal must be willing to truthfully report the realized signal. Intuitively, if multiple signals were to induce a proposal, the principal could report the one that is most probable to induce the agent's acceptance, thereby making the test untrustworthy; meanwhile, combining signals that induce the principal to forgo the project does not affect the equilibrium outcome. It is worth noting that \Cref{prop:binary} is driven by the principal's limited commitment rather than by the recommendation principle \citep[e.g.,][]{bergemann2019information}, which would only allow us to focus on tests with up to $2^{|\types|}$ signals where $|\types|$ is the cardinality of the agent's type space.  
 
Building on the cornerstone result of \Cref{prop:binary}, \Cref{prop:exist} characterizes the optimal test, which takes the simple form of a deterministic binary test. This characterization further delivers two insights. First, even though learning is costless for the principal, fully learning the state may be strictly suboptimal (\Cref{coro:nofulllearn}). Intuitively, committing to learning less could be beneficial because it relaxes her truth-telling constraint in the communication stage. Second, the principal may benefit from reduced payoffs in some states, as it weakens her incentive to misreport upon seeing the null signal (\Cref{coro:overturn}). 


\Cref{sec:prefalign} studies two special cases commonly observed in reality. The case of \textit{positively aligned preferences} arises when players share the same ordinal ranking of the states but differ in their cardinal valuations. For instance, both a CEO and the board prefer a more profitable business plan, and yet they may disagree on the level of profitability that justifies investment. Similarly, both a leading researcher and her potential collaborator favor projects with stronger publication prospects, but the leading researcher may be willing to proceed as long as the project has a chance in a decent journal, while the collaborator may find it worthwhile only if it can hit a top-tier outlet. The case of \textit{negatively aligned preferences}, by contrast, occurs when players' ordinal rankings of the states are opposite: a state that is more desirable for one player is necessarily less desirable for the other. In the political parties' example, if the unknown state represents the prospective policy's electoral side effect, then an effect that benefits one party necessarily harms the other. In the manager–worker setting, the worker prefers an easy task, whereas the manager may favor a difficult one. 

To analyze these two cases, we order the states according to the principal's preference. \Cref{thm:positivealign} shows that under positively aligned preferences, the optimal test is a \textit{threshold test}, under which the project is proposed if and only if the state exceeds a certain threshold.\footnote{The optimality of threshold tests is also discovered by some existing papers on linear persuasion with a privately informed agent and a continuous space (e.g., \cite{kolotilin2022censorship}). However, their results cannot be directly applied here, as the players' payoffs are non-linear in our setting.} Importantly, this result is true even though the principal is uncertain about the level of alignment. It offers plausible explanations for the prevalence of simple pass-or-fail tests in settings such as CEO-board interactions and research collaborations. \Cref{thm:negativealign} shows that under negatively aligned preferences, the optimal test is an \textit{interval test}, which induces a proposal when the state is intermediate, if the agent is relatively more risk-averse than the principal, or a \textit{tail test}, which induces a proposal when the state is extreme, if the agent is relatively more risk-seeking. 

Since the agent privately knows his type, it is natural to ask whether the principal can benefit from screening by offering a menu of tests for the agent to choose from. \Cref{sec:private_menu} provides a positive answer based on a canonical specification of payoff functions. \Cref{thm:private} characterizes the optimal test menu, which includes a threshold test plus up to two interval tests. Notably, the optimal test menu in our setting is much simpler than that in \cite{candogan2023optimal}, a paper that, like ours, highlights the value of screening in information design problems with a privately informed receiver.\footnote{The optimal menu in our paper consists of up to three tests regardless of the cardinality of the agent's type space, and each test takes a simple form. By contrast, the number of tests in the optimal menu in \cite{candogan2023optimal} grows linearly in the number of agent types, and the tests exhibit a laminar structure.} The simplicity of our result arises from the principal's limited commitment, which, again, underscores the value of \Cref{prop:binary}. 

\Cref{sec:extension} presents two extensions. First, we generalize the agent's action space to any arbitrary $\mathcal{A}$, capturing real-world situations where the agent’s decision upon a proposal is richer than simple acceptance or rejection. While a recommendation policy may require up to $|\mathcal{A}|^{|\types|}$ realizable signals, we show that the principal's limited commitment still enables us to focus on binary tests under some mild assumptions on the players' payoff functions. Second, we study a multiple-agent extension in which the principal needs to reach unanimous consent with all the agents. We transform this problem into a single-agent problem in which the principal needs to persuade only the pivotal agent---the one who is least likely to accept the proposal. We find that all the results established for the baseline model continue to hold here. 

\subsection{Related Literature}
\label{sub:literature}

First, this paper contributes to the study of how information design facilitates collective decision-making that requires unanimous consent.\footnote{There are also papers studying the impact of exogenously dispersed information in collective decision-making under the majority rule rather than the unanimity rule. See, e.g., \cite{ali2025political}.} This question has been studied in contexts such as voting \citep{bardhi2018modes} and two-sided matching \citep{xu2023}. In these two papers, the information designer merely controls the information environment. By contrast, in our paper, the information designer (i.e., the principal) is also a party involved in the collective decision.\footnote{Another distinction is that both \cite{bardhi2018modes} and \cite{xu2023} feature multiple agents, while our baseline model has only one agent. To better demonstrate the difference, \Cref{sub:multi-agent} extends our model to multiple agents.} This enables our paper to speak to the applications we are interested in.

Second, this paper relates to the literature on information design with limited commitment in the sense that the designer has no ex-post commitment to truthful reporting of the realized signal.\footnote{In some other papers, it is the \textit{non-designer} who can manipulate the signal. For example, \cite{perez2022test, perez2023fraud} study the test design problem where test-takers can manipulate the input of the test. Other examples include \cite{ball2025scoring}, \cite{li2024screening}, etc.} For instance, in \cite{guo2021costly} and \cite{nguyen2021bayesian}, the designer can misreport at a cost that depends on the distance between the reported signal and the actual one. 
In \cite{lin2024credible}, the designer can misreport as long as the distribution of the reported signals remains identical to that of the actual ones. Unlike those papers, the designer (i.e., the principal) in our model can misreport \textit{at no cost} and \textit{subject to no restrictions}. 
This feature is also seen in \cite{lipnowski2022persuasion}\footnote{In \cite{lipnowski2022persuasion}, the designer can misreport with some exogenous probability. Our setting is aligned with its special case with that probability being one.} and in some papers on overt information acquisition before a cheap-talk game (\citealp{lyu2022information}; \citealp{kreutzkamp2024persuasion}).\footnote{\cite{kreutzkamp2024persuasion} shows that models of overt information acquisition before a cheap-talk game can be re-framed as information design models with the sender's limited commitment. See also papers that study information acquisition before contracting where the principal cannot commit to how to use the acquired information (e.g., \citealp{clark2024optimal}).} 
The crucial distinction between these papers and ours is that the information designer (the principal) in our model is an \textit{active participant} in the post-communication decision-making process --- she has a state-dependent preference and retains veto power over the project, instead of merely providing information. Another important distinction is that our model allows the agent to be privately informed.

\section{Model}
\label{sec:model}


A principal (``she'') and an agent (``he'') can launch a project under unanimous consent. The players' payoffs from the project are governed by a state $\state \in \states:= [-1, 1]$. The state is drawn from a continuous distribution $\statedist(\cdot)$ with probability density function $\statedensity(\cdot)$.\footnote{We assume the distribution  $\statedist(\cdot)$ is atomless to simplify the exposition. All the insights of the paper remain intact for distributions with point masses under slight modifications of the presentation.}

If the project is not launched, both players' payoffs are normalized to zero. If the project is launched, the principal's payoff is $u(\state) \in \mathbbm{R}$, and the agent's is $v(\state, \type)  \in \mathbbm{R}$, where $\type \in \types \subseteq \mathbbm{R}$ is the agent's type, drawn from a distribution $\typedist(\cdot)$ independent of $\statedist(\cdot)$. The agent's type indicates his preference for the project; in particular, $v(\state, \type)$ strictly increases in $\type$, so a higher-type agent gains more from the project. We assume that $\type$ is the agent's private information. Notably, our setting encompasses the special case in which the principal also knows $\type$ --- this occurs if we take $\typedist(\cdot)$ to be a degenerate distribution whose support consists of only the actual value of~$\type$. For technical reasons, both $u(\state)$ and $v(\state, \type)$ are assumed to be bounded. To make the analysis non-trivial, we also assume that $\inf_\state [u(\state)] < 0 < \sup_\state [u(\state)]$ and $\inf_\state [v(\state, \type)] < 0 < \sup_\state [v(\state, \type)]$ for any $\type$.

Neither player observes $\state$ directly. However, the principal can conduct a \textit{test} to acquire some information about $\state$ at no cost. A test $t = (\signals, \test)$ consists of (1) a compact signal space $\signals$ and (2) a measurable mapping $\test: \states \rightarrow \Delta(\signals)$ from the state space to the set of distributions on the signal space. With slight abuse of notation, we let $\test_{\signal}(\state)$ denote the probability of signal $\signal$ given state $\state$ under the mapping $\test$. We let $\tests$ denote the collection of all possible tests. As is standard in the literature, we can also formulate a test as the posterior distribution it induces; that is, we can view a test as an element of $\Delta(\Delta(\states))$ subject to the Bayesian plausibility condition. In our analysis, we will use the two formulations interchangeably.



The game proceeds as follows:
\begin{itemize}
\item \textbf{(Learning stage)} The principal chooses a test $t \in \tests$. Then, the test result $\signal\in\signals$ is generated by $t$ and privately revealed to the principal.
\item \textbf{(Communication stage)} The principal sends a cheap-talk message $\talkmessage \in \talkmessages$ to the agent, where $\talkmessages$ is any measurable set that includes $2^\signals$ as a subset. 
\item \textbf{(P-decision stage)} The principal chooses whether to \textit{propose} or \textit{forgo} the project. If she forgoes it, the project is not launched, and the game ends. 
\item \textbf{(A-decision stage)} If the principal has proposed the project, the agent chooses whether to \textit{accept} the proposal (and launch the project) or \textit{reject} it. 
\end{itemize}

We make two comments about the timeline. First, following the literature on overt information acquisition before cheap-talk communication (e.g., \citealp{kreutzkamp2024persuasion}), we can, without loss of generality, consider a specific message space $\talkmessages = \signals$ in the communication stage. In other words, the communication stage can be reformulated as the principal's reporting a test result $\tilde \signal \in \signals$. Note that the reported test result $\tilde \signal$ may differ from the actual one $\signal$. This reflects our main departure from standard information design models: the principal still has ex-ante commitment to an information structure but no longer has ex-post commitment to truthful reporting of the realized information. Second, the principal chooses a single test instead of offering a menu of tests for the agent to choose from. Since the agent privately knows his type, offering a menu of tests may serve the purpose of screening. To address this possibility, we analyze in \Cref{sec:private_menu} an alternative timeline in which the principal offers a test menu at the beginning of the game for the agent to choose from. 

\subsection{Strategies and Solution Concept}
\label{sec:strt}

In our game, players make decisions in the four named stages identified in the timeline. In the learning stage, the principal's strategy is given by her choice of $t \in \tests$. 

The communication and P-decision stages essentially occur at the same time. In these two stages, the principal's strategy is given by a mapping $\rho_p: \tests\times \signals \rightarrow \Delta(\talkmessages \times \{0,1\})$, where $\rho_p(t, \signal)$ is the joint distribution of the message and whether a proposal is made when she adopts a test $t$ and receives a signal $\signal$. Note that in these two stages, the principal forms a belief about the state $\state$, denoted by $\principalPosterior \in \Psi_p = \Delta(\states)$. 

Finally, in the A-decision stage, the agent's strategy is given by $\rho_a: \tests\times \talkmessages \rightarrow [0,1]$, where $\rho_a(t, \talkmessage)$ is the probability that the agent accepts the proposal when the principal adopts the test $t$, sends the message $\talkmessage$, and makes a proposal. In this stage, the agent has a belief about $\state$, which we denote by $\agentPosterior \in \Psi_a = \Delta(\states)$. 

We now specify the solution concept. Given the test $t$ chosen by the principal, we define an \textit{equilibrium} of the continuation game as a tuple $(\rho_p, \rho_a, \principalPosterior, \agentPosterior)$ that satisfies the following conditions. Essentially, our equilibrium concept augments weak Perfect Bayesian equilibrium with an additional requirement about off-equilibrium-path pessimism, as reflected by the fourth condition.\footnote{In fact, the requirement on off-equilibrium-path belief only serves to simplify the analysis and is not essential for the paper. This is because we will select the principal-preferred equilibrium, as mentioned later. All results continue to hold if we use instead the concept of weak perfect Bayesian equilibrium (WPBE), where we keep only the first three conditions.}
\begin{enumerate}
    \item Each player's strategy is sequentially rational given their own belief and the other player's strategy. 
    \item The principal's belief $\principalPosterior$ is formed by Bayes' rule.\footnote{Note that the principal will not face any off-equilibrium-path scenario, as she finishes her actions before the agent's action starts.}
    \item If there is a positive probability that the project is proposed and the message $\talkmessage$ is sent, the agent's belief $\agentPosterior$ is formed by Bayes' rule. 
    \item If there is zero probability that the project is proposed and the message $\talkmessage$ is sent, the agent's belief $\agentPosterior$ is arbitrarily set as any state that induces him a negative payoff, leading him always to reject the proposal.\footnote{Such a state exists since we assume $\inf_\state [v(\state, \type)] < 0$ for any $\type$ in the model.}
\end{enumerate}

Finally, we select the principal-preferred equilibrium when there are multiple equilibria. This is a common practice to rule out the babbling equilibrium in cheap-talk models. 

\section{Limited Commitment Makes Simplicity}
\label{sec:general}


Following the literature on overt information acquisition before cheap-talk communication (e.g., \citealp{kreutzkamp2024persuasion}), we can, without loss of generality, restrict attention to \textit{truthful equilibria} where (a) the message space is $M=S$, and (b) the principal finds it incentive compatible to truthfully report the test result. On the equilibrium path of a truthful equilibrium, the principal and the agent form the same posterior belief about $\state$, which is induced by the realized result $\signal$ from the test $t$ and denoted by $\posterior_t(\signal) \in \Delta(\states)$. 

We introduce a tie-breaking assumption that the agent always accepts the proposal when being indifferent. Let $\rho_t(\tilde{\signal},\type) \in \{0,1\}$ denote a type-$\type$ agent's acceptance decision if the principal adopts the test $t$, reports $\tilde{\signal}$, and makes a proposal. If the reported signal $\tilde{\signal}$ may induce the principal's proposal on the equilibrium path, the type-$\type$ agent accepts the proposal (i.e., $\rho_t(\tilde{\signal},\type)=1$) if and only if $\expect[\state\sim \posterior_t(\tilde{\signal})]{v(\state,\type)} \geq 0$.\footnote{In the off-equilibrium-path scenario where the principal proposes the project with a message that would not be accompanied by a proposal (i.e., a signal that would induce her to forgo the project), we have $\rho_t(\tilde{\signal},\type)=0$ according to our equilibrium concept.}


Upon seeing the actual signal $\signal$, the principal finds it optimal to propose the project if and only if $\expect[\state\sim \posterior_t(\signal)]{u(\state)} \geq 0$. If this condition does not hold, her expected payoff is zero regardless of what signal she reports. If it holds, her expected payoff from reporting the signal $\tilde\signal$ is $\expect[\state\sim \posterior_t(\signal)]{u(\state)} \cdot \prob{\rho_t(\tilde{\signal},\type) =1}$. Therefore, we let 
\begin{equation*}
\intP_t(\tilde{\signal}, \signal)=\expect[\state\sim \posterior_t(\signal)]{u(\state)} 
\cdot \prob{\rho_t(\tilde{\signal},\type) =1}
\cdot \mathbf{1}\{\expect[\state\sim \posterior_t(\signal)]{u(\state)} \geq 0\} 
\end{equation*}
be the principal's expected payoff from adopting the test $t$, receiving the signal $\signal$, reporting the signal~$\tilde{\signal}$, and making the optimal proposal decision based on her belief.

A necessary condition for such an equilibrium to exist is that the test must be trustworthy in the sense that the principal does not want to misreport the test result, as formally defined below.
\begin{definition}
\label{def:trustworthy}
A test $t=(\signals, \test)$ is \textbf{trustworthy} if $\intP_t(\signal, \signal) \geq \intP_t(\tilde{\signal}, \signal)$ for any $\tilde{\signal}, \signal \in \signals$.
\end{definition}

We can now formalize the principal's optimization problem as
\begin{problem-nonum}[\textbf{P}]\label{problem:P}  
~\\ \vspace{-55pt}
\begin{align}
\max_{t \in \tests} \quad & \expect[\signal]{\intP_t(\signal, \signal)} \nonumber \\ %
\text{\rm s.t.} \quad & \intP_t(\signal, \signal) \geq \intP_t(\tilde{\signal}, \signal), \ \forall \tilde\signal\neq\signal.  \tag{IC-P}\label{icp}
\end{align}
\end{problem-nonum}
Notably, its only departure from classic information design problems is the additional constraint (\ref{icp}), which stems from the principal's limited commitment in our setting. Without this constraint, finding the optimal test would be challenging because the agent has a private type, the state is continuous, and the players' payoffs are non-linear in the state --- a complex setting of information design where no solution is known even for binary actions.\footnote{Previous information design models that feature a privately informed receiver and a continuous state space focus primarily on linear persuasion problems (e.g., \cite{kolotilin2022censorship,candogan2023optimal}). In these problems, the sender's payoff depends only on the expected state, standing in contrast to our setting.} However, it turns out that the constraint (\ref{icp}), which represents limited commitment, drastically simplifies the analysis and makes the problem tractable, as indicated by the following theorem.


\begin{theorem}[Limited commitment makes simplicity]
\label{prop:binary}
It is without loss of optimality for the principal to adopt a \textbf{binary test} (i.e., a test with two realizable signals) and propose the project only under one of the two signals.\footnote{Our definition of binary tests includes the degenerate case with only one signal (i.e., the test is completely uninformative). This case can be viewed as letting the probability of sending the other signal be zero.} 
\end{theorem}
\begin{proof}
    See \Cref{pf:prop:binary}. 
\end{proof}

\Cref{prop:binary} is driven by the principal's limited commitment. Suppose, for the sake of contradiction, that the test has at least two distinct realizable signals, say $\signal', \signal'' \in \signals$, under which the principal will propose the project. Then the principal will benefit from either misreporting $\signal'$ as $\signal''$ or vice versa---she will always report the signal that induces more agent types to accept the proposal. Hence, such a test is not trustworthy. On the other hand, if there are two realizable signals under which the principal will forgo the project, without loss of optimality, they can be combined into a single signal. 


It is worth emphasizing that \Cref{prop:binary} is \textit{not} driven by the recommendation principle (see, e.g., \citealp{bergemann2019information}). Specifically, a general recommendation policy corresponds to a mapping from the agent's type space to his action space, which may require up to $2^{|\Lambda|}$ distinct realizable signals, where $|\Lambda|$ is the cardinality of the agent's type space.\footnote{Indeed, \Cref{prop:binary} holds more generally when the agent's action space is more complex. In \Cref{sub:multi-action}, we generalize the agent's action space to any arbitrary $\mathcal{A}$. While a recommendation policy may require up to $|\mathcal{A}|^{|\Lambda|}$ distinct realizable signals, the principal's limited commitment still enables us to focus on binary tests under some mild assumptions on the players' payoff functions.} Hence, the recommendation principle allows us to focus on binary tests only in the degenerate case where $|\Lambda| = 1$ (i.e., the agent's type is publicly known). By contrast, the limited-commitment-driven logic underlying \Cref{prop:binary} applies more generally to situations where the agent's type is private. 

We see the value of \Cref{prop:binary} as twofold. First, it paves the way for solving an information design problem that would otherwise be intractable under full commitment, enabling us to generate some insights on how learning and communication can facilitate unanimous consent when the principal faces uncertainty about the agent's preference. Second, it provides a plausible explanation for why tests adopted in real practice typically take a simple pass-or-fail form, even in complex situations where the agent privately knows his preference.\footnote{For instance, a typical consulting report on a business plan concludes with a clear recommendation on whether the plan is worth pursuing. Likewise, a think tank report assessing a policy often delivers a bottom-line judgment on whether the policy is worth moving forward with.} This stands in contrast to some earlier papers on persuading a privately informed receiver, where the characterized optimal test is rather complex, albeit delicate (see, e.g., \cite{candogan2023optimal}, which we will discuss with further details in \Cref{sec:private_menu}).

Building on \Cref{prop:binary}, the next theorem characterizes the optimal test. Hereafter, given a binary test that the principal adopts, we refer to the signal that induces her to propose the project as the \textit{proposal signal} and the other one as the \textit{null signal}. 

\begin{theorem}[Existence and characterization]
\label{prop:exist}
There exists an optimal test, represented by a set of states $\states_1^* \subseteq \states$. The test generates the proposal signal if $\state \in \states_1^*$ and the null signal if $\state \in \states_0^* := \states/\states_1^*$. The set $\states_1^*$ solves the following problem.\footnote{In this problem, the randomness associated with $\prob{\int_{\states_1} v(\state,\type) \dd\statedist(\state) \geq 0}$ comes from $\type$.} 
\begin{problem-nonum}[\textbf{P*}]\label{problem:Pstar}  
~\\ \vspace{-33pt}
\begin{align}
\max_{\states_1\subseteq \states} \quad & \int_{\states_1} u(\state) \dd\statedist(\state) \cdot \prob{\int_{\states_1} v(\state,\type) \dd\statedist(\state) \geq 0}  \nonumber \\
\text{\rm s.t.} \quad & \int_{\states/\states_1} u(\state) \dd\statedist(\state) \leq 0 \quad \text{ if } \states_1\neq \emptyset.   \tag{IC-P'} \label{icp2} \nonumber
\end{align}
\end{problem-nonum}
\end{theorem}
\begin{proof}
See \Cref{pf:prop:exist}. 
\end{proof}

\Cref{prop:exist} shows the optimality of a \textit{deterministic binary test} (i.e., one that does not randomize between two signals for any $\state$). This test dichotomizes the state space into a set of \textit{proposal states} $\states_1^*$ and a set of \textit{null states} $\states_0^*$. The objective function of the embedded optimization problem represents the principal's expected payoff from proposing the project upon seeing the proposal signal, given the agent's best response. The constraint (\ref{icp2}), adapted from (\ref{icp}), guarantees that the principal does not want to propose the project upon seeing the null signal.\footnote{We do not consider the other direction of misreporting because a test yielding the principal a negative payoff under the proposal signal will never be optimal for this problem, as the principal has the option to choose $\states_1 = \emptyset$, which secures her a zero expected payoff.}

The characterization result in \Cref{prop:exist} gives rise to two corollaries. 

\begin{corollary}
\label{coro:nofulllearn}
The principal finds it optimal to commit to ``learning less.'' 
\end{corollary}

While learning is costless for the principal, the characterized optimal test in \Cref{prop:exist} is a coarse test that does not fully reveal the state. The principal benefits from committing to ``learning less,'' as it relaxes her truth-telling constraint in the cheap-talk communication stage. Indeed, fully learning the state may be strictly sub-optimal for the principal. If the principal fully learns the state, her optimal continuation strategy is to propose the project in the states $\states_1^{p}:=\{\state\given u(\state)\geq 0\}$ and forgo it in the remaining states. This is analogous to using a binary test with proposal states $\states_1^{p}$ in our model. Hence, fully learning the state is strictly sub-optimal as long as $\states_1^*$ outperforms $\states_1^{p}$.\footnote{Some existing papers make a similar point that the sender in cheap-talk communication may benefit from being imperfectly informed, but they are based on different settings from ours. See, e.g., \cite{fischer2001imperfect}, \cite{ivanov2010informational}.}

\Cref{coro:nofulllearn} also suggests the value of ex-ante commitment in the absence of ex-post commitment. In an alternative game where the principal does not have both types of commitment, it is in her interest to fully learn the state, and the reasoning in the previous paragraph follows. Hence, the value of the principal's ex-ante commitment can be quantified by the difference in her payoff between adopting $\states_1^*$ versus $\states_1^{p}$.

Before stating the next corollary, we define an \textit{environment} as a tuple $\mathcal{E}=\{u, v\}$ that includes the players' payoff functions. Given an environment $\mathcal{E}$, we let $\states_1^*(\mathcal{E})$ represent the proposal states of an optimal deterministic binary test, $\states_0^*(\mathcal{E}):=\states/\states_1^*(\mathcal{E})$ the null states, and $\mathcal{U}(\mathcal{E})$ the principal's expected payoff under the optimal test. 


\begin{corollary}
\label{coro:overturn}
The principal may benefit from reduced payoffs in some states. Formally, let $\mathcal{E}=\{u, v\}$ and $\tilde{\mathcal{E}}=\{\tilde u, \tilde v\}$ satisfy $v=\tilde v$. If $\tilde u(\state) \leq u(\state)$ for $\state \in \states_0^*(\mathcal{E})$ and $\tilde u(\state) = u(\state)$ for $\state \in \states_1^*(\mathcal{E})$, then $\mathcal{U}(\tilde{\mathcal{E})} \geq \mathcal{U}(\mathcal{E})$. 
\end{corollary}
\begin{proof}
See \Cref{pf:coro:overturn}.     
\end{proof}

As suggested by \Cref{coro:overturn}, given an environment $\mathcal{E}$, reducing the principal's payoffs in the states $\states_0^*(\mathcal{E})$ turns out to (weakly) benefit the principal. Intuitively, $\states_0^*(\mathcal{E})$ are the states that tend to induce the principal to forgo the project, so reduced payoffs in these states weaken her incentive to misreport upon seeing the null signal, thereby relaxing her truth-telling constraint and potentially improving her welfare. This comparative statics result contrasts sharply with the conventional wisdom that reducing state-wise payoffs can only harm the principal.


Although \Cref{coro:overturn} features weak inequalities, \Cref{fig:strictimprove} provides one example where the principal \textit{strictly} benefits from reduced payoffs in some states. In this example, the state $\state$ follows a uniform distribution over $[-1,1]$, and the agent has only one type. In the environment $\mathcal{E}$, the principal's payoff is $1+2\epsilon$ if $\state\geq 0$ and $-1$ if $\state < 0$; the agent's payoff is $-1-\epsilon$ if $\state \geq 0$ and $1$ if $\state < 0$. If the proposal states $\states_1^*(\mathcal{E}) \neq \emptyset$, it must include strictly more negative states compared to positive ones in order to make the agent accept. However, this indicates that the corresponding null states, $\states_0^*(\mathcal{E})$, must include strictly more positive states compared to negative ones, which makes the principal benefit from misreporting upon seeing the null signal, rendering the test untrustworthy. Hence, it must hold that $\states_1^*(\mathcal{E}) = \emptyset$, indicating that the project can never be launched in this environment, and the principal's expected payoff is $\mathcal{U}(\mathcal{E}) = 0$. 

Now, we switch to the environment $\tilde{\mathcal{E}}$, whose only departure from $\mathcal{E}$ is the reduction of the principal's payoff from $1+2\epsilon$ to $1-2\epsilon$ for $\state\in(\frac{1}{2}, 1]$. In this new environment, the principal benefits from using a binary test whose proposal states are $\states_1^* = [-\frac{1+\epsilon}{2}, \frac{1}{2}]$. Under this test, the agent is willing to accept the proposal, and more importantly, the principal's expected payoff from the null states, $\states_0^* = [-1, -\frac{1+\epsilon}{2}) \cup (\frac{1}{2}, 1]$, becomes strictly negative, deterring her from misreporting upon seeing the null signal. Indeed, we can show that it is the principal's optimal test, yielding her a positive expected payoff of $\mathcal{U}(\tilde{\mathcal{E}}) = \frac{\epsilon}{4} > \mathcal{U}(\mathcal{E})$. 

\begin{singlespace}
\begin{figure}[htbp]
    \begin{center}
	\subfigure[Environment $\mathcal{E}$]{
		\begin{minipage}[t]{0.45\textwidth}
        \centering
		\begin{tikzpicture}[xscale=0.22, yscale = 0.38]
						\draw [thick, ->] (-11,-5.5) -- (-11,6.5);
			\draw [thick, ->] (-11,0) -- (11,0) node [right] {\scriptsize  $\theta$};
			\draw [-] (-10,0.1) -- (-10,-0.1) node [below] {\scriptsize $-1$};
			\draw [-] (10,0.1) -- (10,-0.1) node [below] {\scriptsize $1$};
			\draw [-] (0,0.1) -- (0,-0.1) node [below] {\scriptsize $0$};
			\node [left] at (-11,3) {\scriptsize  $1$};
			\node [left] at (-11,-3) {\scriptsize  $-1$};
			\node [left] at (-11,-4) {\scriptsize  $-1-\epsilon$};
			\node [left] at (-11,5) {\scriptsize  $1+2\epsilon$};

			\draw [ultra thick, blue] (0,5) -- (10,5) node [right] {\scriptsize  $u(\theta)$};
			\draw [ultra thick, red, dashed] (0,-4) -- (10,-4) node [right] {\scriptsize  $v(\theta, \type)$};
			\draw [ultra thick, blue] (0,-3) -- (-5,-3);
			\draw [ultra thick, red, dashed] (0,3) -- (-10,3);
            \draw [ultra thick, blue] (-10,-3) -- (-5,-3);

			\draw [dotted] (-11,5) -- (0,5) -- (0,0);
			\draw [dotted] (-11,3) -- (-10,3) -- (-10,0);
			\draw [dotted] (10,5) -- (10,0);
			\draw [dotted] (-11,-3) -- (-10,-3) -- (-10,-1.5);
			\draw [dotted] (10,-4) -- (10,-1.5);
			\draw [dotted] (-11,-4) -- (0,-4) -- (0,-1.5);
            
		\end{tikzpicture}    
		\end{minipage}
	}
	\subfigure[Environment $\tilde{\mathcal{E}}$]{
		\begin{minipage}[t]{0.45\textwidth}
       \centering
		\begin{tikzpicture}[xscale=0.22, yscale = 0.38]
						\draw [thick, ->] (-11,-5.5) -- (-11,6.5);
			\draw [thick, ->] (-11,0) -- (11,0) node [right] {\scriptsize  $\theta$};
			\draw [-] (-10,0.1) -- (-10,-0.1) node [below] {\scriptsize $-1$};
			\draw [-] (10,0.1) -- (10,-0.1) node [below] {\scriptsize $1$};
            			\draw [-] (5,0.1) -- (5,0) node [below] {\scriptsize $\frac{1}{2}$};
			\draw [-] (0,0.1) -- (0,-0.1) node [below] {\scriptsize $0$};
			\node [left] at (-11,3) {\scriptsize  $1$};
			\node [left] at (-11,-3) {\scriptsize  $-1$};
			\node [left] at (-11,5) {\scriptsize  $1+2\epsilon$};
			\node [left] at (-11,1) {\scriptsize  $1-2\epsilon$};
            \node [left] at (-11,-4) {\scriptsize  $-1-\epsilon$};

			\draw [ultra thick, blue] (0,5) -- (5,5);
			\draw [ultra thick, blue] (5,1) -- (10,1) node [right] {\scriptsize  $\tilde{u}(\theta)$};
			\draw [ultra thick, red, dashed] (0,-4) -- (10,-4) node [right] {\scriptsize  $\tilde{v}(\theta, \type)$};
			\draw [ultra thick, blue] (0,-3) -- (-5,-3);
			\draw [ultra thick, red, dashed] (0,3) -- (-10,3);
            \draw [ultra thick, blue] (-10,-3) -- (-5,-3);

			\draw [dotted] (-11,5) -- (0,5) -- (0,0);
			\draw [dotted] (-11,3) -- (-10,3) -- (-10,0);
			\draw [dotted] (10,1) -- (10,0);
			\draw [dotted] (-11,-3) -- (-10,-3) -- (-10,-1.5);
            \draw [dotted] (-11,1) -- (5,1) -- (5,5);
            \draw [dotted] (5,0) -- (5,1);
			\draw [dotted] (10,-4) -- (10,-1.5);
			\draw [dotted] (0,-3) -- (0,-1.5);
			\draw [dotted] (-11,-4) -- (0,-4) -- (0,-1.5);
            
		\end{tikzpicture}    
		\end{minipage}
        }	
    \end{center}
\caption{Illustration of the example where the principal strictly benefits from reduced payoffs in some states. The state $\state$ follows a uniform distribution over $[-1,1]$. In each panel, the blue solid curve and the red dashed curve depict the principal's and the agent's payoff functions, respectively.}
\label{fig:strictimprove}
\end{figure}
\end{singlespace}

\section{Preference Alignment}
\label{sec:prefalign}

The previous section directs our attention to deterministic binary tests. In this section, we derive sharper predictions for two prevalent special cases in which players' preferences over the states are either positively or negatively aligned. 

Positive alignment of preferences arises when players share the same ordinal ranking of the states but may differ in their cardinal valuations. Two examples from the introduction illustrate this case. In the CEO–board setting, both parties prefer a more profitable business plan, yet they may disagree on the level of profitability that warrants investment. Similarly, in the researchers' example, both the leading researcher and the potential collaborator favor projects with stronger publication prospects, but it may be the case that the leading researcher is willing to proceed as long as the project has a chance in a decent journal, whereas the collaborator may find it worthwhile only if it can hit a top-tier outlet.

Negative alignment of preferences, by contrast, occurs when players' ordinal rankings of the states are opposite: a state that is more desirable for one player is necessarily less desirable for the other. The remaining two examples in the introduction fall into this category. In the political parties' example, if the unknown state represents the prospective policy's side effect in an upcoming election, then an effect that benefits one party necessarily harms the other. In the manager–worker setting, the worker prefers an easy task, whereas the manager may benefit from a difficult one.

To study these two cases, we normalize the principal's payoff to be $u(\state)=\state$ in this section. In other words, the states are ordered by the principal's preference. 
\begin{definition}
\label{def:align}
The players' preferences are \textbf{positively aligned} if $v(\state, \type)$ is increasing in $\state$ for any $\type$ and \textbf{negatively aligned} if $v(\state, \type)$ is decreasing in $\state$ for any $\type$. 
\end{definition}

For these two cases that are prevalent in reality, it would be challenging to find the principal's optimal test in a standard information design problem where she had full commitment, particularly due to the combination of the agent's private type, the continuity of the state space, and the non-linearity of the players' payoff functions.\footnote{Specifically, many existing results on linear persuasion problems (e.g., \cite{kolotilin2022censorship}) cannot apply to our setting, even in the two special cases in \Cref{sec:prefalign}.} However, owing to limited commitment in our setting, which directs us to focus on binary tests, the problem becomes tractable. 

\begin{definition}\label{def:monotone_test}
(a) A \textbf{threshold test} is a deterministic binary test whose set of proposal states $\states_1^*$ is $[\hat{\state},1]$ for some partition threshold $\hat{\state}$. \\ 
(b) An \textbf{interval test} is a deterministic binary test whose set of proposal states $\states_1^*$ is $[\underline{\state}, \bar{\state}]$ for some $\underline{\state} < \bar{\state}$. \\
(c) A \textbf{tail test} is a deterministic binary test whose set of proposal states $\states_1^*$ is $[-1, \underline{\state}]\cup[\bar{\state},1]$ for some $\underline{\state} < \bar{\state}$. 
\end{definition}


Notice that a threshold test can be viewed both as a degenerate interval test with $\bar{\state}=1$ and as a degenerate tail test with $\underline{\state}<-1$.\footnote{In the definition of tail tests, we allow for the degenerate cases of $\underline{\state} < -1$, where the left tail is excluded, and of $\bar{\state}>1$, where the right tail is excluded.} 

\begin{theorem}[Positive alignment]
\label{thm:positivealign}
If players' preferences are positively aligned, the optimal test is a threshold test with $\states_1^*=[\state^*, 1]$, where $\state^*\geq 0$. 
\end{theorem}
\begin{proof}
See \Cref{pf:thm:positivealign}. 
\end{proof}

The intuition behind \Cref{thm:positivealign} is as follows. If we disregard the agent's acceptance decision, the principal's first-best strategy is to learn whether $\state\in[0,1]$ and propose the project whenever it is the case. This strategy, however, may work poorly in inducing acceptance when there is a high probability that the agent's expected payoff from $\state\in[0,1]$ is negative. In such situations, the principal benefits from excluding $[0,\state^*)$ from $\states_1^*$, where the threshold $\state^*$ is optimally chosen to balance the gain from inducing more agent acceptance against the loss from forgoing the project when it could have benefited the principal. 

\Cref{thm:positivealign} establishes the optimality of threshold tests in situations with positively aligned preferences. Importantly, this conclusion remains true even when the principal is uncertain about the level of preference alignment. This is consistent with observations in the two examples discussed above. In practice, a CEO may seek to persuade the board that a business plan's profit exceeds a critical threshold, despite being uncertain about the board's precise acceptance criteria. Similarly, a researcher may attempt to convince a potential collaborator that a project meets some standards, even without knowing the collaborator's exact rules for project screening. 

\begin{theorem}[Negative alignment]
\label{thm:negativealign}
Suppose players' preferences are negatively aligned. \\
(a) If $v(\state, \type)$ is concave in $\state$ for any $\type$, the optimal test is either $\states_1^*=\emptyset$ or an interval test with $\states_1^*=[\state^*, \state^{**}]$, where $\state^* \leq 0 \leq \state^{**}$. \\
(b) If $v(\state, \type)$ is convex in $\state$ for any $\type$, the optimal test is either $\states_1^*=\emptyset$ or a tail test with $\states_1^*=[-1, \state^{**}]\cup[\state^*, 1]$, where $\state^{**} \leq \state^{*}$. 
\end{theorem}
\begin{proof}
See \Cref{pf:thm:negativealign}. 
\end{proof}

Intuitively, under negatively aligned preferences, the principal's problem is to identify the states that optimally balance the dual objectives of generating a positive payoff for herself and inducing the agent's acceptance. When the agent is relatively more risk-averse than the principal,\footnote{Our notion of risk attitude concerns how a player values mean-preserving spreads in the posterior belief about $\state$, differing from the classic notion, which concerns how one values mean-preserving spreads in monetary lotteries.} reflected by the concavity of his payoff function (case (a)), intermediate states are the best option to propose, leading to the optimality of an interval test. For instance, in the manager-worker context, where the worker may disproportionately dislike excessively difficult tasks, this result implies that the principal tends to propose a task whose level of difficulty is certified to lie within a manageable band. 

By contrast, when the agent is relatively more risk-seeking, reflected by the convexity of his payoff function (case (b)), extreme states become more effective, directing the principal toward a tail test. In the example of political persuasion, where the involved parties may disfavor intermediate states,\footnote{For example, a moderate outcome may disproportionately benefit a third party that is not involved in the persuasion.} this finding suggests the theoretical possibility that consensus is more likely to be reached on policies that always favor one of the parties involved in the persuasion.

\section{Value of Screening}
\label{sec:private_menu}

Since the agent privately knows his type, a natural question is whether the principal can benefit from screening by offering a menu of tests for the agent to choose from. This section provides a positive answer. 

To illustrate such a possibility, we consider a special case where the principal's payoff is $u(\state) = \state$, and the agent's is $v(\state, \type) = \type - \state$. We assume $\type > 0$ to ensure that there exist some states mutually beneficial to both players. 
\begin{definition}
\label{def:exante}
In this special case, a type-$\type$ agent is \textbf{ex-ante optimistic} if $\type \geq \expect{\state}$ and \textbf{ex-ante pessimistic} if $\type < \expect{\state}$. 
\end{definition}
This definition specifies two scenarios that are distinguished by whether the agent prefers to launch the project based \textit{solely} on his prior belief about the state. In other words, if the principal \textit{always} proposes the project regardless of the test result, an ex-ante optimistic agent will accept the proposal while an ex-ante pessimistic one will reject it.

We begin by characterizing the optimal single test when the principal cannot perform screening, as in the baseline model. This will serve as the benchmark. 

\begin{proposition}
\label{thm:private_single}
Suppose the principal can only use a single test. \\
(a) If $\prob{\type \geq \expect{\state}} = 0$, then the optimal test is $\states_1^*=\emptyset$. \\
(b) If $\prob{\type \geq \expect{\state}} > 0$, then the optimal test is a threshold test.
\end{proposition}
\begin{proof}
    See Appendix~\ref{pf:thm:private_single}. 
\end{proof}

The reason behind \Cref{thm:private_single}(a) is that the project can never be launched if the agent is \textit{always} ex-ante pessimistic. Intuitively, given the players' negatively aligned preferences, a proposal from the principal can only intensify the agent's pessimism about the state. Hence, if the agent is ex-ante pessimistic, he will become even more pessimistic upon receiving any proposal and thus reject it. To put this formally, an agent's posterior mean belief about $\state$ upon receiving a proposal must exceed the prior mean $\expect{\state}$, which, in part (a), is already assumed to be higher than $\type$ with certainty. 

By contrast, \Cref{thm:private_single}(b) indicates that the principal can benefit from adopting a test if the agent is ex-ante optimistic with positive probability. Since this special case features negatively aligned preferences with a linear agent payoff function, which is both convex and concave, the optimal test ends up being a threshold test, which is simultaneously a degenerate interval test (see \Cref{thm:positivealign}) and a degenerate tail test (see \Cref{thm:negativealign}), as pointed out following \Cref{def:monotone_test}.

Next, we present the optimal menu of tests if the principal can perform screening. 

\begin{theorem}[Value of screening]
\label{thm:private}
Suppose the principal can use a menu of tests. \\
(a) If $\prob{\type \geq \expect{\state}} = 0$, then it is optimal to offer a singleton test $\states_1^*=\emptyset$. \\
(b) If $\prob{\type \geq \expect{\state}} > 0$, then it is optimal to offer a menu consisting of up to three tests, including one threshold test plus zero, one, or two interval tests. 
\end{theorem}
\begin{proof}
    See Appendix~\ref{pf:thm:private}. 
\end{proof}

Notice that introducing screening does not alter the principal's lack of ex-post commitment in the communication stage. Hence, the reasoning of \Cref{prop:binary} continues to work in this section, enabling us to focus on test menus that only consist of binary tests. As a consequence, \Cref{thm:private}(a) holds for the same reason as \Cref{thm:private_single}(a) --- an ex-ante pessimistic agent only becomes more pessimistic upon the principal's proposal, and thus the project is never launched under any menu of tests.

By contrast, \Cref{thm:private}(b) characterizes the optimal menu of tests when the agent is ex-ante optimistic with positive probability. The optimal menu includes a threshold test, which targets high-type agents; in addition, the principal may benefit from offering one or two interval tests targeting lower-type agents. Compared to the threshold test, these interval tests induce the principal's proposal with lower probability, but upon her proposal, the expected state is also more favorable to the agent. Therefore, a higher-type agent, who is less sensitive to the state, prefers the threshold test, whereas a lower-type agent, who is more sensitive to the state, will choose an interval test. \Cref{sec:example} provides an example to illustrate how screening is achieved.\footnote{The example in \Cref{sec:example} has two tests in the optimal menu. To demonstrate the tightness of our characterization in \cref{thm:private}, in \cref{apx:three_test_menu}, we provide a more complicated example where the optimal menu contains three tests.}

The most notable feature of the characterized optimal menu is its simplicity, which is particularly striking given that both the agent's type space $\types$ and the state space $\states$ are continuous. To highlight this point, we compare our result with earlier work on persuading a privately informed receiver that also emphasizes the value of screening.\footnote{There are also papers finding that screening does not benefit the sender in some specific settings, such as \citet{guo2019interval} and \citet{kolotilin2017persuasion}.} One such example is \cite{candogan2023optimal}, which shows that screening can be helpful in linear persuasion problems.\footnote{Relatedly, \citet{hancart2024optimal} also underscores the value of screening, but based on a different rationale. In \citet{hancart2024optimal}, the principal lacks access to a fully-revealing information structure, prompting the use of a test menu to extract additional information about the \textit{payoff-relevant state} through the agent's test choice. By contrast, the principal in our model can fully learn the relevant information, but wants to learn about the agent's \textit{preference type} via screening.} In their paper, the optimal test menu exhibits a laminar structure, with the number of tests growing linearly with the number of receiver types. By contrast, our optimal menu consists of at most three tests regardless of the richness of the agent's type space --- even when it is continuous. Moreover, each test takes a particularly simple form, namely a threshold test or an interval test.

It is worth emphasizing that the simplicity of our characterization is owing to our restriction to binary tests, which is driven by the principal's limited commitment. As demonstrated by the proof, the restriction to binary tests allows us to convert the problem to a linear optimization problem with three constraints, whose extreme points correspond to test menus that consist of up to three deterministic binary tests. This once again underscores the value of \Cref{prop:binary} in rendering an otherwise complex problem tractable and in yielding a solution that is simple and realistic.

\subsection{An Example Where Screening Is Helpful}
\label{sec:example}

Consider the following example. The state $\state$ is drawn from a uniform distribution on $[-1,1]$. The agent has two possible types, $\type_1 = \frac{2}{3}$ and $\type_2 = \frac{1}{3}$. The prior probabilities of these types are $q_1=1-\epsilon$ and $q_2=\epsilon$, respectively, where $\epsilon>0$ is a sufficiently small constant. 

In this example, the optimal menu consists of two tests: (1) a threshold test $t_1$ whose associated proposal states are $[0, 1]$, and (2) an interval test $t_2$ whose associated proposal states are $[\frac{1}{12}, \frac{7}{12}]$. With the test $t_1$, the project is proposed with probability $p_1 = \frac{1}{2}$, and the expected state upon proposal is $\posteriormean_1 = \frac{1}{2}$. With the test $t_2$, the project is proposed with probability $p_2 = \frac{1}{4}$, and the expected state upon proposal is $\posteriormean_2 = \frac{1}{3}$. It is not difficult to verify that it is incentive-compatible for a $\type_1$-type agent to choose $t_1$ and a $\type_2$-type agent to choose $t_2$. Moreover, as long as $\epsilon$ is sufficiently small, this test menu enables the principal to achieve the highest possible expected payoff, as the project yields her a non-negative payoff if and only if $\state \geq 0$. 

For comparison, suppose the principal can use only a single test in this example. \cref{thm:private_single} implies that the optimal single test is a threshold test. Hence, the principal may choose either the threshold $-\frac{1}{3}$ (so that the induced expected state upon proposal is $\frac{1}{3}$) to induce both agent types to accept the proposal or the threshold $0$ to serve only the high type. When $\epsilon$ is sufficiently small, it is optimal to adopt the latter threshold, which exactly corresponds to the test $t_1$. 

Apparently, the two-test menu outperforms the optimal single test. In particular, the addition of the interval test $t_2$ enables the project to be launched, yielding the principal a positive payoff, with positive probability when the agent's type is low. 

\section{Discussions}
\label{sec:extension}

\subsection{Richer Action Space}
\label{sub:multi-action}

In the baseline model, the agent faces a binary decision upon the principal's proposal. In many real-world settings, however, the agent's choice is richer than simple acceptance or rejection. For example, in the CEO-board context, the board may decide not only whether to invest but also how much to invest. To capture such situations, we generalize the model by enriching the agent's action space. 

Formally, we extend the baseline model by letting the agent choose an action $a \in \mathcal{A} \subseteq [0, \infty)$ upon the principal's proposal. We require $0 \in \mathcal{A}$, where the action $a=0$ corresponds to rejection in the baseline model; additionally, we let a larger action represent a higher level of agent engagement with the project. The principal's payoff takes the form of $\tilde{u}(\state, a) = u(\state) \cdot a$, so that she benefits from the agent's engagement if and only if she finds the state favorable (i.e., $u(\state)>0$). The agent's payoff becomes $\tilde{v}(\state, \type, a)$, satisfying $\tilde{v}(\state, \type, 0) = 0$, $\forall \state, \type$. We further assume that $\tilde{v}(\state, \type, a)$ (i) strictly increases and is linear in $\state$ and (ii) is super-modular in $\state$ and $a$.\footnote{One example of the agent's payoff function is $\tilde{v}(\state, \type, a) = (\state+\type)\cdot a - \frac{a^2}{2}$.} Because of its linearity in $\state$, the agent's optimal action depends on the mean of his posterior belief about the state; we let $a^*(\mu, \type)$ denote a type-$\type$ agent's optimal action when his posterior mean is $\mu$. Because of its monotonicity in $\state$ and super-modularity in $\state$ and $a$, we can infer that $a^*(\mu, \type)$ weakly increases in $\mu$ for any $\type$. 

In this generalized setting, the recommendation principle only allows us to focus on tests with up to $|\mathcal{A}|^{|\Lambda|}$ realizable signals, where $|\mathcal{A}|$ and $|\Lambda|$ are the cardinalities of the agent's action space and type space, respectively. However, owing to the principal's limited commitment, we can continue to focus on binary tests as in the baseline model. 
\begin{proposition}
\label{prop:binary_multi}
When the agent has a richer action space, it is without loss of optimality for the principal to adopt a \textbf{binary test} and propose the project only under one of the two realizable signals. 
\end{proposition}
\begin{proof}
See \Cref{pf:prop:binary_multi}. 
\end{proof}

The intuition behind this proposition is similar to \Cref{prop:binary}. In particular, if multiple distinct realizable signals induce the principal to propose the project, the principal has an incentive always to report the signal that induces the highest level of agent engagement. 

\subsection{Multiple Agents}
\label{sub:multi-agent}

In many situations, a collective action requires the unanimous consent of more than two parties; therefore, the party with learning and communication privilege needs to convince multiple other parties simultaneously. We thus consider an alternative setting with multiple agents indexed by $i\in \mathcal{I}:=\{1, 2, \ldots, n\}$. The project is launched only when the principal proposes it and all the agents accept it. If the project is launched, the principal's payoff is still $u(\state)$, while each agent $i$'s payoff becomes $v(\state, \type_i)$ with $\type_i$ following a distribution $\typedist_i(\cdot)$. Since the principal's limited commitment is still in place in this extension, \Cref{prop:binary} continues to work. Hence, the principal's problem is to find the binary test that maximizes 
\begin{align}
\max_{\states_1\subseteq \states} \quad & \int_{\states_1} u(\state) \dd\statedist(\state) \cdot \prob{\int_{\states_1} v(\state,\type_i) \dd\statedist(\state) \geq 0, \forall i \in \mathcal{I}}  \nonumber\\
\text{s.t.} \quad & \int_{\states_0} u(\state) \dd\statedist(\state) \leq 0.  \nonumber
\end{align}

As the agent's utility is monotone in $\type$, a lower-type agent's acceptance implies a higher-type's. Thus, the principal's problem is to convince the \textit{pivotal agent}, the one with the lowest realized type. Because $\type_i$ is random, the pivotal agent's type $\type_{\text{pivotal}}:=\min_{i \in \mathcal{I}} \type_i$ follows a distribution $\typedist_{\text{pivotal}}(\cdot)$, which is determined by the distributions $\{\typedist_i(\cdot)\}_{i \in \mathcal{I}}$. So, we transform the problem into one where the principal faces a \textit{single} agent whose type distribution is $\typedist_{\text{pivotal}}(\cdot)$, and all the analysis in the baseline model carries over to this extension.

\section{Concluding Remarks}

This paper studies how a leading party can acquire and communicate information to achieve unanimous consent with other parties. We focus on situations where the leading party faces limited commitment: she can commit ex-ante to an information structure but lacks ex-post commitment to truth-telling. 

The paper offers two main insights. Methodologically, we discover the possibility for limited commitment to render some otherwise intractable information design problems tractable, facilitating the analysis of settings that are not well understood under full commitment. Substantively, we provide a plausible explanation for the prevalence of simple binary tests in practice, even when the agent's type space and action space are highly complex. 

Several directions are promising for further exploration. For example, it would be intriguing to extend the analysis to alternative decision rules, such as majority rule.\footnote{Some existing papers study the effects of dispersed information on collective decision-making under majority rule, but take the information environment as exogenous. See, e.g., \cite{ali2025political}.} Another valuable extension is to introduce multiple principals who may compete in persuasion, a setting that is also analytically challenging under full commitment.

\newpage

\appendix

\section{Proofs}
\label{apx:proofs}

\subsection{Proof of \Cref{prop:binary}}
\label{pf:prop:binary}

We prove this by contradiction. Suppose the principal adopts a test $t = (\signals, \test)$ such that there are two signals $\signal',\signal'' \in \signals$ under which she proposes the project. Under the assumption that $v(\state, \type)$ increases in $\type$ for any $\state$, the approval set is a cutoff set.  Suppose the signals $\signal'$ and $\signal''$ induce the agent to accept the project if and only if his preference type $\type$ is weakly above $\type'$ and $\type''$, respectively. If $\type' = \type''$, we can combine $\signal'$ and $\signal''$ into a single signal without affecting the equilibrium outcome. If $\type' < \type''$ (i.e., $\signal'$ induces more agent types to accept the proposal), the principal will benefit from misreporting $\signal''$ as $\signal'$, rendering the test untrustworthy. Similarly, $\type' > \type''$ cannot hold, either. 


Finally, if there are multiple realizable signals under which the principal forgoes the project, we can combine them into one signal without loss of generality. 

\subsection{Proof of \Cref{prop:exist}}
\label{pf:prop:exist}


\Cref{prop:binary} allows us to focus on binary tests. Let $\omega(\state)\in[0,1]$ be the probability of the proposal signal if the state is $\state$. We can rewrite Problem~(\nameref{problem:P}) as 

\begin{problem-nonum}[\textbf{P$_1$}]\label{problem:P1}  
~\\ \vspace{-33pt}
\begin{align}
\max_{\omega} \quad & \int_{\states} u(\state)\omega(\state) \dd\statedist(\state) \cdot \prob{\int_{\states} v(\state,\type)\omega(\state) \dd\statedist(\state) \geq 0} \nonumber  \\ 
\text{s.t.} \quad & \int_{\states} u(\state) (1-\omega(\state)) \dd\statedist(\state) \leq 0 \quad \text{ if } \omega(\state) \not\equiv 0.  \nonumber
\end{align}
\end{problem-nonum}
Next, we show the theorem in two steps. 

\paragraph{\underline{Step 1.}} We show that it is without loss of optimality for the principal to adopt a deterministic binary test (i.e., one that does not randomize between two signals for any $\state$). We prove this argument by considering two cases.

\paragraph{Case 1.} If the proposal is rejected by all agent types in an optimal solution, then the optimal solution can be written as a deterministic binary test, $\omega^*(\state)\equiv 0, \forall \state$. 

\paragraph{Case 2.} If the proposal is accepted by some agent types in an optimal solution $\omega^*$, then we know that $\omega^*(\state)\not\equiv 0$. Let $\type^*$ denote the smallest agent type to accept under the solution $\omega^*$. Then, $\omega^*$ must also solve
\begin{problem-nonum}[\textbf{P$_2$}]\label{problem:P2}  
~\\ \vspace{-55pt}
\begin{align}
\max_{\omega} \quad & \int_{\states} u(\state)\omega(\state) \dd\statedist(\state) \cdot \prob{\type\geq \type^*} \nonumber  \\ 
\text{\rm s.t.} \quad & \int_{\states} u(\state) (1-\omega(\state)) \dd\statedist(\state) \leq 0  \nonumber \\
\quad & \int_{\states} v(\state,\type^*)\omega(\state) \dd\statedist(\state) \geq 0, \nonumber 
\end{align}
\end{problem-nonum}
Otherwise, any solution $\hat{\omega}$ to Problem~(\nameref{problem:P2}) will also outperform $\omega^*$ in Problem~(\nameref{problem:P1}), as it already outperforms $\omega^*$ even when not accounting for the potential acceptance of agent types strictly below $\type^*$. 

\begin{lemma}
\label{lemma:bangbang}
If Problem~(\nameref{problem:P2}) admits a solution, it must also admit a bang-bang solution (i.e., a solution that takes values only in $0$ and $1$).
\end{lemma}
\begin{proof}
In Problem~(\nameref{problem:P2}), $\type^*$ is taken as given, so the objective is equivalent to maximizing $\int_{\states} u(\state)\omega(\state) \dd\statedist(\state)$. The first constraint is equivalent to $\int_{\states} u(\state) \omega(\state) \dd\statedist(\state) \geq \expect{u(\state)}$, whose left-hand-side is identical to the objective and the right-hand-side is a constant. Hence, if a solution to the \textit{relaxed problem}, defined as Problem~(\nameref{problem:P2}) without the first constraint, violates the first constraint, then Problem~(\nameref{problem:P2}) should admit no solution, violating the presumption of this lemma. Therefore, any solution to Problem~(\nameref{problem:P2}) must also solve
\begin{problem-nonum}[\textbf{P$_3$}]\label{problem:P3}  
~\\ \vspace{-55pt}
\begin{align}
\max_{\omega} \quad & \int_{\states} u(\state)\omega(\state) \dd\statedist(\state)  \nonumber \\ 
\text{\rm s.t.} \quad & \int_{\states} v(\state,\type^*)\omega(\state) \dd\statedist(\state) \geq 0, \nonumber 
\end{align}
\end{problem-nonum}
\noindent whose Lagrangian functional is 
$$\mathcal{L} = \int_{\states} [u(\state)+\eta\cdot v(\state,\type^*)]  \omega(\state) \dd\statedist(\state).$$
If a solution to Problem~(\nameref{problem:P3}), which we denote by $\omega^*$, is not a bang-bang solution, let $Z^*:= \{\state\given\omega^*(\state)\in(0,1) \}$. We know that $u(\state)+\eta^* \cdot v(\state,\type^*) = 0$ for any $\state \in Z^*$, where $\eta^*$ is the Lagrangian multiplier in the optimal solution. Now, we want to construct an alternative $\omega^{**}$ that is a bang-bang solution to Problem~(\nameref{problem:P3}). 

If the set $Z^*$ has zero measure, then the construction is immediate, as we only need to boost the value of $\omega^*(\state)$ to 1 for the states in $Z^*$. If the set $Z^*$ has a positive measure, it suffices to find an alternative set $Z^{**} \subseteq Z^*$ such that $\int_{Z^{**}} v(\state,\type^*) \dd\statedist(\state) = \int_{Z^*} v(\state,\type^*)\omega^*(\state) \dd\statedist(\state)$. Because the distribution $\statedist$ is atomless, such a $Z^{**}$ exists. In particular, we construct $\omega^{**}$ as 
$$\omega^{**}(\state) = \begin{cases}
    \mathbf{1}(\state\in Z^{**}) & \text{ if } \state\in Z^* \\
    \omega^*(\state) &  \text{ if } \state\notin Z^*.
\end{cases}$$
It is a bang-bang solution that also solves Problem~(\nameref{problem:P3}).
\end{proof}
  
Notably, this lemma is mainly driven by the assumption that the state distribution $\statedist(\cdot)$ is atomless. This lemma implies that the constructed $\omega^{**}$ works as well as $\omega^*$ in Problem~(\nameref{problem:P2}) and, most importantly, can be represented by a deterministic binary test. From its construction, we can infer that the smallest agent type to accept under $\omega^{**}$ is also $\type^*$, and therefore, we conclude that $\omega^{**}$ also solves Problem~(\nameref{problem:P1}). 

Therefore, we conclude Step 1. The purpose of this step is to simplify the principal's problem to identifying a set of \textit{proposal states} $\states_1 \subseteq \states$, which is exactly Problem~(\nameref{problem:Pstar}) featured in the theorem and copied below.  
\vspace{-6pt}
\begin{align}
\max_{\states_1\subseteq \states} \quad & \int_{\states_1} u(\state) \dd\statedist(\state) \cdot \prob{\int_{\states_1} v(\state,\type) \dd\statedist(\state) \geq 0}  \nonumber \\
\text{s.t.} \quad & \int_{\states/\states_1} u(\state) \dd\statedist(\state) \leq 0 \quad \text{ if } \states_1\neq \emptyset,   \tag{IC-P'} \nonumber
\end{align}


\paragraph{\underline{Step 2.}} We show that a solution to Problem~(\nameref{problem:Pstar}) exists. 

Since the function $u(\state)$ is bounded and the distribution $G(\cdot)$ is atomless, we know that $\int_{\states_1} u(\state) \dd\statedist(\state)$ is continuous in $\states_1$. Since the function $v(\state, \type)$ is bounded and we assume that the agent always accepts the proposal when being indifferent, we know that $\prob{\int_{\states_1} v(\state,\type) \dd\statedist(\state) \geq 0}$ is upper semi-continuous in $\states_1$. Putting these two arguments together, the objective function is upper semi-continuous in $\states_1$. 

Meanwhile, since we require the signal generation rule $\test$ to be a measurable mapping, we can infer that in a deterministic binary test, the set of proposal states $\states_1$ must also be measurable. This allows us to confirm that the feasible set (the set of $\states_1$ satisfying \eqref{icp2}) must be compact because (a) $\int_{\states/\states_1} u(\state) \dd\statedist(\state)$ is continuous in $\states_1$ due to a bounded $u(\state)$ and an atomless $G(\cdot)$, and (b) \eqref{icp2} features a weak inequality.

Finally, by Weierstrass' extreme value theorem (semi-continuity version), a maximizer for Problem~(\nameref{problem:Pstar}) must exist because the feasible set is compact and the objective function is upper semi-continuous in $\states_1$. This proves the existence of an optimal test in the form of a deterministic binary test. 

\subsection{Proof of \Cref{coro:overturn}}
\label{pf:coro:overturn}


Notice that $\states_1^*(\mathcal{E})$ is still a feasible solution in the environment $\tilde{\mathcal{E}}$, as $\tilde{u}(\state)\leq u(\state)$ for $\state \in \states_0^*(\mathcal{E})$. Hence, in $\tilde{\mathcal{E}}$, the principal can do weakly better than adopting $\states_1^*(\mathcal{E})$, which already yields her a payoff equal to $\mathcal{U}(\mathcal{E})$. 




\subsection{Proof of \Cref{thm:positivealign}}
\label{pf:thm:positivealign}

As pointed out in the proof of \cref{lemma:bangbang}, any solution $\omega^*$ to the principal's problem must also solve Problem~(\nameref{problem:P3}) for some $\type^*$, which we copy below. 
\begin{align}
\max_{\omega} \quad & \int_{\states} \state\cdot\omega(\state) \dd\statedist(\state)  \nonumber \\
\text{s.t.} \quad & \int_{\states} v(\state,\type^*)\omega(\state) \dd\statedist(\state) \geq 0. \nonumber 
\end{align}
Its Lagrangian functional is 
$$\mathcal{L} = \int_{\states} [\state+\eta\cdot v(\state,\type^*)]  \omega(\state) \dd\statedist(\state).$$
Since the Lagrangian multiplier $\eta\geq 0$, when the players' preferences are positively aligned, the term $\state+\eta\cdot v(\state,\type^*)$ is strictly increasing in $\state$, and therefore, the maximizer of the Lagrangian functional must take the form of $\omega^*(\state) = \indicate{\state \in [\state^*,1]}$, where $\state^*$ satisfies $\state^*+\eta^* \cdot v(\state^*,\type^*) = 0$ with $\eta^*$ being the Lagrangian multiplier in the optimal solution. 

We also want to show that $\state^*\geq 0$. By contradiction, if $\state^*<0$, it follows that $v(\state^*,\type^*) > 0$, as $\state^*+\eta^* \cdot v(\state^*,\type^*) = 0$ must hold. However, this further indicates that the type $\type^*$ receives a strictly positive payoff from accepting the proposal since $v(\state, \type^*) \geq v(\state^*,\type^*) > 0$ for $\state > \state^*$. Hence, the constraint in the above optimization problem is non-binding, suggesting $\eta^* = 0$. This, in turn, indicates that $\state^*+\eta^* \cdot v(\state^*,\type^*) = \state^* < 0$, a contradiction. 
 
\subsection{Proof of \Cref{thm:negativealign}}
\label{pf:thm:negativealign}

Similar to the proof of \Cref{thm:positivealign}, a solution $\omega^*$ to Problem~(\nameref{problem:P3}) must satisfy 
\begin{equation*}
\omega^*(\state) = \begin{cases}
    1 \quad \text{  if  } \state + \eta^* \cdot v(\state,\type^*)\geq 0 \\
    0 \quad \text{  if  } \state + \eta^* \cdot v(\state,\type^*) < 0, 
    \end{cases}
\end{equation*}
for some Lagrangian multiplier $\eta^*\geq 0$ and minimal agent type of acceptance $\type^*$. 

\paragraph{\underline{Part (a).}} When $v(\state, \type)$ is concave, the term $\state + \eta^* \cdot v(\state,\type^*)$ is also concave. Hence, either (a) this term is always negative, which corresponds to $\states_1^*=\emptyset$, or (b) it is (weakly) positive if and only if $\state$ belongs to an interval $[\state^*, \state^{**}]$. For the latter case, it remains to show that $\state^* \leq 0 \leq \state^{**}$. Notice that $\state^{**} < 0$ is impossible, as it implies a strictly negative payoff for the principal to propose the project. Meanwhile, if, by contradiction, $\state^* > 0$, then it follows that $\state^* - \epsilon + \eta^* \cdot v(\state^* - \epsilon, \type^*) < 0$ for an arbitrarily small $\epsilon$, indicating that $v(\state^* - \epsilon,\type^*) < 0$. Since $v(\state, \type)$ decreases in $\state$, we can further infer that $v(\state,\type^*) < 0$ for $\state \geq \state^*$. This contradicts the type-$\type^*$ agent's incentive to accept the proposal when $\state\in[\state^*, \state^{**}]$. 

\paragraph{\underline{Part (b).}} When $v(\state, \type)$ is convex, the term $\state + \eta^* \cdot v(\state,\type^*)$ is also convex. Hence, either (a) this term is always negative, which corresponds to $\states_1^*=\emptyset$, or (b) it is (weakly) positive if and only if $\state$ belongs to some $[-1, \state^{**}] \cup [\state^{*}, 1]$.

\subsection{Proof of \Cref{thm:private_single}}
\label{pf:thm:private_single}

In \Cref{sec:private_menu}, both players' payoff functions are linear in $\state$. Hence, the relevant features of a deterministic binary test, represented by its proposal states $\states_1$, are (a) the probability of the proposal signal, which we denote by $p(\states_1)$, and (b) the expected state induced by the proposal signal, which we denote by $\posteriormean(\states_1)$.\footnote{For notational consistency, we allow $\posteriormean(\states_1)$ to be any number in $\states$ when $\states_1=\emptyset$.} In other words, two binary tests $\states_1$ and $\states_1'$ result in the same equilibrium payoffs for the players if $p(\states_1)=p(\states_1')$ and $\posteriormean(\states_1)=\posteriormean(\states_1')$. 

Accordingly, we let $\tilde{p}(\states_1):=1- p(\states_1)$ denote the probability of the null signal and $\tilde{\posteriormean}(\states_1) = \frac{\expect{\state} - p(\states_1) \cdot \posteriormean(\states_1)}{1-p(\states_1)}$ denote the expected state induced by the null signal. Since the principal's payoff is normalized to be $u(\state)=\state$ in this case, a deterministic binary test $\states_1$ is trustworthy if and only if $\tilde{\posteriormean}(\states_1) \leq 0 \leq \posteriormean(\states_1)$. This suggests that $\posteriormean(\states_1) \geq \tilde{\posteriormean}(\states_1)$, which further implies a necessary condition, $\posteriormean(\states_1) \geq \expect{\state}$, because by Bayes' rule, the posterior mean belief about $\state$ upon realization of the proposal signal should be weakly higher than the prior mean. 

Since $p(\states_1)$ and $\posteriormean(\states_1)$ capture the essential information of a deterministic binary test $\states_1$, we henceforth shift the analysis from the space of $\states_1$ to the space of $(p,\posteriormean)$. We require $\posteriormean\geq \expect{\state}$ (the necessary condition derived in the last paragraph) and say that a pair $(p,\posteriormean)$ can be induced by a deterministic binary test if there exists $\states_1$ such that $p(\states_1)=p$ and $\posteriormean(\states_1)=\posteriormean$.

\begin{lemma}\label{lem:feasible_posterior_mean}
For $\mu \in [\expect{\state}, 1]$, let $\state_\posteriormean$ be the only value that satisfies $\expect{\state\given\state\geq \state_\posteriormean} = \posteriormean$. \\
(a) The pair $(p, \posteriormean)$ can be induced by a deterministic binary test if and only if $p \leq 1-\statedist(\state_{\posteriormean})$.\\
(b) If $p = 1-\statedist(\state_{\posteriormean})$, then the pair $(p, \posteriormean)$ can only be induced by a threshold test (up to difference of a zero-measured set). \\
(c) If $p < 1-\statedist(\state_{\posteriormean})$, then the pair $(p, \posteriormean)$ cannot be induced by a threshold test. Instead, it can be induced by an interval test. 
\end{lemma}
\begin{proof}
To prove this lemma, we establish the following arguments. First, if $p = 1-\statedist(\state_{\posteriormean})$, by the definition of $\state_\posteriormean$, the threshold test $\states_1=[\state_\posteriormean, 1]$ induces the pair; from the construction, we can also infer that the pair induced by a threshold test must satisfy $p = 1-\statedist(\state_{\posteriormean})$. 

Second, if $p < 1-\statedist(\state_{\posteriormean})$, we can construct an interval test to induce the pair. Based on the threshold test $\states_1=[\state_\posteriormean, 1]$, we simultaneously remove some proposal states from both tails of the interval $[\state_\posteriormean, 1]$ so that the probability of proposal is reduced to $p$ while the expected state remains unchanged at $\posteriormean$. This construction gives us an interval test that induces $(p, \posteriormean)$. 

Third, if $p > 1-\statedist(\state_{\posteriormean})$, we show that the pair cannot be induced by any deterministic binary test. Specifically, for any deterministic binary test $\states_1$ with $p(\states_1) = p$, the maximal expected state is achieved by the threshold test $[G^{-1}(1-p), 1]$. For this test, the expected state satisfies 
$$\expect{\state\given \state\geq G^{-1}(1-p)} < \expect{\state\given \state\geq \state_{\posteriormean}} = \posteriormean,$$
where the inequality follows from $p > 1-\statedist(\state_{\posteriormean})$. Hence, no deterministic binary tests with proposal probability $p$ can reach an expected state of $\posteriormean$ if $p > 1-\statedist(\state_{\posteriormean})$. 
\end{proof}

We are now ready to show the proposition. 

\paragraph{\underline{Statement (a).}} Suppose, by contradiction, there exists an optimal test $\states_1^* \neq \emptyset$. A type-$\type$ agent's mean posterior belief about $\state$ upon seeing the proposal signal must satisfy 
$$\posteriormean(\states_1^*) \geq \expect{\state} > \type,$$
where the first inequality exploits the necessary condition identified earlier in the proof, and the second inequality follows from the premise of statement (a). Hence, the agent should not accept any proposal based on his posterior belief. 

\paragraph{\underline{Statement (b).}} Instead of directly applying Theorems~\ref{thm:positivealign} and \ref{thm:negativealign}, we provide a more intuitive proof for statement (b). We want to show that, given a trustworthy deterministic binary test that does not admit a threshold form, there is a trustworthy threshold test that strictly improves the principal's expected payoff. 

Suppose the principal uses a binary non-threshold test $\states_1^*$. The fact that it does not admit a threshold form implies that $p(\states_1^*) < 1-\statedist(\state_{\posteriormean(\states_1^*)})$ according to \Cref{lem:feasible_posterior_mean}. Moreover, for it to be trustworthy, $\tilde{\posteriormean}(\states_1^*) \leq 0 \leq \posteriormean(\states_1^*)$ must hold. 

Now, we construct an alternative threshold test $\states_1^{'}:=[\state_{\posteriormean(\states_1^*)}, 1]$. Under such construction, we have $\posteriormean(\states_1^{'})=\posteriormean(\states_1^*)$, and according to \Cref{lem:feasible_posterior_mean}, $p(\states_1^{'}) = 1-\statedist(\state_{\posteriormean(\states_1^*)}) > p(\states_1^{*})$. Additionally, we have $\tilde{p}(\states_1^{'}) < \tilde{p}(\states_1^{*})$ and $\tilde{\posteriormean}(\states_1^{'}) = \frac{\expect{\state} - p(\states_1^{'}) \posteriormean(\states_1^{'})}{\tilde{p}(\states_1^{'})} < \frac{\expect{\state} - p(\states_1^{*}) \posteriormean(\states_1^{*})}{\tilde{p}(\states_1^{*})} = \tilde{\posteriormean}(\states_1^{*})$. 

Hence, this constructed threshold test (i) leads to the same decision on the part of the agent because $\posteriormean(\states_1^{'})=\posteriormean(\states_1^*)$, (ii) remains trustworthy because $\tilde{\posteriormean}(\states_1^{'}) < \tilde{\posteriormean}(\states_1^{*}) \leq 0 \leq \posteriormean(\states_1^{*}) = \posteriormean(\states_1^{'})$, and (iii) strictly increases the principal's expected payoff because its proposal probability is strictly larger than the original test while the expected state induced by the proposal signal remains unchanged.

\subsection{Proof of \Cref{thm:private}}
\label{pf:thm:private}

Notice that allowing the principal to offer a menu of tests does not alter her limited commitment, and consequently, the reasoning in \Cref{sec:general} continues to hold --- we can restrict attention to menus containing only deterministic binary tests. Hence, \Cref{thm:private}(a) holds for the same reason as in \Cref{thm:private_single}. No matter what test is adopted, the principal's proposal always intensifies the agent's pessimism about the state, so the project can never be launched if the agent is always ex-ante pessimistic. The rest of this appendix is dedicated to proving \Cref{thm:private}(b). 

Following the notation of \Cref{pf:thm:private_single}, we derive the optimal menu of tests in the space of $(p, \posteriormean)$ instead of $\states_1$. By the revelation principle, we can restrict attention to test menus that induce the agent to report his type truthfully. Hence, we can formulate the principal's problem as choosing $\{p(\type), \posteriormean(\type)\}_{\type\in\types}$, where $(p(\type), \posteriormean(\type))$ is induced by the deterministic binary test designed for a type-$\type$ agent. 

Several constraints need to be satisfied. First, the pair $(p(\type), \posteriormean(\type))$ must be \textit{feasible} in the sense that it can be induced by a deterministic binary test, i.e., $p(\type) \leq 1-\statedist(\state_{\posteriormean(\type)})$ for any $\type$ according to \Cref{lem:feasible_posterior_mean}. We will refer to this constraint as \eqref{fea}. 

Second, the test that induces $(p(\type), \posteriormean(\type))$ must be \textit{trustworthy}. This is satisfied if and only if $\tilde{\posteriormean}(\type) \leq 0 \leq \posteriormean(\type)$, which is equivalent to $\posteriormean(\type) \geq \max\left\{0, \frac{\expect{\state}}{p(\type)}\right\}$. We will refer to this constraint as \eqref{obm}, as it is adapted from the constraint \eqref{icp} specified in Problem~(\nameref{problem:P}). 

Third, the agent must be willing to report his true type. Let a type-$\type$ agent's expected payoff from fully mimicking the behavior of a type-$\type'$ agent (i.e., from choosing the test $(p(\type'), \posteriormean(\type'))$ and deciding whether to accept the proposal in the same way as a type-$\type'$ agent) by $\intutil(\type,\type') = [\type - \posteriormean(\type')]  p(\type')$. This constraint requires $\intutil(\type,\type) \geq \intutil(\type,\type')$ for any $\type\neq\type'$, which we refer to as \eqref{ica}.\footnote{\label{fn:1}As written, \eqref{ica} captures only one deviating strategy of the agent, which is to fully mimic the behavior of another type. Actually, since the agent's decision is binary, the conditions \eqref{ica} and \eqref{obw} combined completely account for all of his deviation strategies. In particular, the only deviation not captured by \eqref{ica} and \eqref{obw} is the double deviation in which he misreports his type and rejects the proposal. However, if this deviation is profitable, then it must also be profitable for him to reject the proposal after truthful reporting, which violates \eqref{obw}.} 

Finally, we also need the constraint \eqref{obw}, as shown below, to ensure that the agent is willing to accept the proposal if proposed. With all the constraints laid out, we can formulate the principal's problem of finding the optimal test menu as follows. 
\begin{problem-nonum}[\textbf{S}] \label{problem:S}
~\\ \vspace{-33pt}
\begin{align}
\max_{\{p(\type),\posteriormean(\type)\}_{\type}} \quad & \expect[\type]{p(\type)\cdot\posteriormean(\type)}  \nonumber\\
\text{s.t.} \quad &\intutil(\type,\type) \geq \intutil(\type,\type'), \ \forall \type'\neq\type,  \tag{IC-A}\label{ica}\\
& \posteriormean(\type) \geq \max\left\{0, \frac{\expect{\state}}{p(\type)}\right\} \quad \text{ if } p(\type) >0, \tag{IC-P''} \label{obm} \\
& \posteriormean(\type) \leq \type \quad \text{ if } p(\type)>0, \tag{IR-A} \label{obw} \\
& p(\type) \leq 1-\statedist(\state_{\posteriormean(\type)}).  \tag{FSB} \label{fea}
\end{align}
\end{problem-nonum}

Problem~(\nameref{problem:S}) can be simplified considerably. If $p(\type)\not\equiv 0$, let $\underline{\type}:= \inf\{\type\given p(\type)>0\}$ and $\bar{\type}:= \sup\{\type\given p(\type)>0\}$ denote the lowest and highest agent types such that the project is launched with positive probability. First, \eqref{obw} is implied by \eqref{ica} because the agent's interim utility must be non-negative. Second, by the standard envelope theorem argument,\footnote{To see how the envelope theorem directly applies, we can regard the agent in our setting as one who privately knows his value of a good $\type$ and can buy the good at the price $\posteriormean(\type')$ with probability $p(\type')$ if reporting his valuation to be $\type'$. See, e.g., \cite{milgrom2002envelope}.} \eqref{ica} is equivalent to $p(\type)$ being weakly increasing in $\type$ and 
\begin{align}
p(\type) \cdot \posteriormean(\type) = p(\type)\cdot \type - \int_{\underline{\type}}^{\type} p(z) \dd z - \mathcal{V}(\underline{\type}), \label{eq:evp}
\end{align}
where $\mathcal{V}(\type):=V(\type, \type)$ is the type-$\type$ agent's expected payoff. Third, the envelope theorem also indicates that $\mathcal{V}(\type)$ is weakly increasing in $\type$, so in the optimal solution, $\mathcal{V}(\underline{\type}) = 0$ must hold, corresponding to $\posteriormean(\underline{\type}) = \underline{\type}$. Fourth, \eqref{eq:evp} implies that $p(\type) \cdot \posteriormean(\type)$ weakly increases in $\type$, and consequently, when combining with the third argument, we can infer that \eqref{obm} is equivalent to $p(\underline{\type})\cdot \underline{\type} \geq \max\{\expect{\state},0\}$. Since we assume $\lambda > 0$, this inequality can be further simplified to $p(\underline{\type}) \geq \frac{\expect{\state}}{\underline{\type}}$. Fifth, since the left-hand side of \eqref{fea} weakly increases in $\type$, and the right-hand side weakly decreases in $\type$, it suffices to consider \eqref{fea} at $\type = \bar{\type}$. Given the value of $p(\bar{\type})$, let $\bar{\posteriormean}$ be the posterior mean such that $p(\bar{\type}) = 1-\statedist(\state_{\bar{\posteriormean}})$. This allows us to rewrite \eqref{fea} as $\posteriormean(\bar{\type}) \leq \bar{\posteriormean}$. Using \eqref{eq:evp} to express $\posteriormean(\bar{\type})$, \eqref{fea} is equivalent to $\bar{\type} - \frac{1}{p(\bar{\type})} \cdot\rbr{\int_{\underline{\type}}^{\bar{\type}} p(z) \dd z} \leq \bar{\posteriormean}$. With these five arguments, the principal's optimal test menu is either $p(\type)\equiv 0$ or the solution to the following problem. 
\begin{problem-nonum}[\textbf{S}'] \label{problem:Sprime}
~\\ \vspace{-33pt}
\begin{align*}
\max_{\{p(\type)\}_{\type}} \quad & \expect[\type]{p(\type)\cdot \type - \int_{\underline{\type}}^{\type} p(z) \dd z} \\
\text{s.t.} \quad & p(\type) \text{ is weakly increasing},\\
& p(\underline{\type}) \geq \frac{\expect{\state}}{\underline{\type}},\\
& \bar{\type} - \frac{1}{p(\bar{\type})} \cdot\rbr{\int_{\underline{\type}}^{\bar{\type}} p(z) \dd z} \leq \bar{\posteriormean}. 
\end{align*}
\end{problem-nonum}

Problem~(\nameref{problem:Sprime}) is a linear optimization problem (when fixing the choice of $p(\bar{\type})$), so it is without loss of optimality to consider the extreme points of the feasible set. In particular, the constraints include a monotonicity constraint, a boundedness constraint, and an integration constraint. Using Proposition 2.1 of \citet{winkler1988extreme}, we know that an extreme point includes up to three distinct non-zero values for $p(\type)$, in addition to the value $0$ that appears when $\type<\underline{\type}$. Moreover, in the optimal solution, \eqref{fea} must bind for $\type=\bar{\type}$; otherwise, the principal can increase $p(\bar{\type})$ and $p(\bar{\type})\cdot\posteriormean(\bar{\type})$ simultaneously to increase her expected payoff without violating any constraint. According to \Cref{lem:feasible_posterior_mean}, this suggests that the test offered to $\bar{\type}$ can be implemented by a threshold test.  

We can now summarize our findings. The corresponding $p(\type)$ of an optimal solution to Problem~(\nameref{problem:S}) is either (a) $p(\type)\equiv 0$ or (b) a step function with at most three non-zero values that has binding \eqref{fea} for $\bar{\type}$. The former case means that the project is never launched, corresponding to \Cref{thm:private}(a). The latter case can be implemented by a threshold test together with up to two interval tests, corresponding to \Cref{thm:private}(b).

\subsection{Proof of \cref{prop:binary_multi}}
\label{pf:prop:binary_multi}

Suppose two distinct realizable signals, say $s'$ and $s''$, induce the principal to propose the project. Let their induced posterior belief be $\posterior'$ and $\posterior''$, and corresponding expected state be $\mu'$ and $\mu''$. Since the principal is willing to propose the project under both signals, we have $\expect[\theta\sim \posterior']{u(\theta)} \geq 0$ and $\expect[\theta\sim \posterior'']{u(\theta)} \geq 0$, and therefore, under both signals, the principal benefits from inducing a larger expected action of the agent. 

Without loss of generality, let $\mu' \geq \mu''$. It follows that $\expect[\type]{a^*(\mu',\type)} \geq \expect[\type]{a^*(\mu'',\type)}$ due to the fact that $a^*(\mu, \type)$ weakly increases in $\mu$ for any $\type$. If the above inequality is strict, then the principal strictly prefers to report $s'$ instead of $s''$ when the true signal is $s''$, rendering the test untrustworthy; otherwise, we can combine the two signals into one without affecting the outcome of the game. 

Finally, if multiple realizable signals induce the principal to forgo the project, we can combine them into one signal without affecting the outcome of the game.

\newpage
\section{Miscellaneous Analysis}

\subsection{An Example of Optimal Menu with Three Tests}
\label{apx:three_test_menu}

Consider the following example. The state $\state$ is drawn from a uniform distribution on $[-\frac{1}{3},\frac{2}{3}]$. The agent has three possible types, $\type_1 = \frac{7}{24}, \type_2 = \frac{1}{2}$, and $\type_3 = \frac{2}{3}$. The prior probabilities of these types are $q_1=(\frac{95}{243}-\delta)\cdot\epsilon, q_2=(\frac{148}{243}+\delta)\cdot\epsilon$, and $q_3=1-\epsilon$, where $\delta>0$ and $\epsilon>0$ are both sufficiently small constants. 

In this example, the optimal menu consists of three tests: $t_1$ is an interval test represented by $(p_1 = \frac{4}{7}, \posteriormean_1 = \frac{7}{24})$; $t_2$ is an interval test represented by $(p_2 = \frac{13}{21}, \posteriormean_2 = \frac{4}{13})$; $t_3$ is a threshold test with threshold $0$, represented by $(p_3 = \frac{2}{3}, \posteriormean_3 = \frac{1}{3})$. We can verify that this test menu satisfies all the constraints of Problem~(\nameref{problem:S}). Moreover, it is optimal when $\delta$ and $\epsilon$ are sufficiently small. 

The intuition of why we need three tests in the optimal menu is as follows. First, when $\epsilon$ is sufficiently small, the type distribution is concentrated around $\type_3$. Hence, the optimal menu should mainly target $\type_3$-type and offer $t_3$ as if it is the only type. 

Second, for the two lower types, the idea is to design tests for them without violating the incentive constraint of the type $\type_3$. The ratio of the probabilities $\frac{q_1}{q_2}$ is set such that it is slightly better to offer a test with a higher proposal probability for $\type_2$ and a lower proposal probability for $\type_1$. Ideally, subject to the agent's incentive constraint, the menu would offer $(\hat{p}_1=\frac{8}{15},\hat{\posteriormean}_1=\frac{7}{24})$ and $(\hat{p}_2=\frac{13}{21},\hat{\posteriormean}_2=\frac{32}{75})$ to the agent. However, this construction violates the principal's incentive. In particular, the principal benefits from misreporting under the test $(\hat{p}_1=\frac{8}{15},\hat{\posteriormean}_1=\frac{7}{24})$, because the posterior mean induced by the null signal is $\frac{1}{42}>0$.

To avoid misreporting, the principal has two options --- she can either offer the three-test menu we mentioned above or remove the test option for the type $\type_1$ (so that the type $\type_1$ does not receive any allocation). When $\delta$ is sufficiently small, the former option provides a higher expected payoff to the principal and is thus optimal.

\end{spacing}

\begin{spacing}{1.2}
\newpage
\bibliographystyle{apalike}
\bibliography{ref}

\end{spacing}

\end{document}